\documentclass[onecolumn]{IEEEtran}
\usepackage{algorithm}
\usepackage{algpseudocode}
\usepackage{times,amssymb,amsmath,amsfonts,nicefrac,float,accents,color,
graphicx,bbm,caption,subcaption,mathrsfs,amsthm,stmaryrd,soul}
\usepackage[noadjust]{cite}
\usepackage{tikz}
\usetikzlibrary{calc}
\usepackage[mathscr]{eucal}

\newtheorem{theorem}{Theorem}[section]
\newtheorem{lemma}[theorem]{Lemma}

\newtheorem{corollary}[theorem]{Corollary}

\newtheorem{definition}[theorem]{Definition}
\newtheorem{construction}[theorem]{Construction}
\newtheorem{remark}[theorem]{Remark}

\newtheorem{example}[theorem]{Example}

\newcommand{\ml}{\mathcal}

\newcommand{\lc}{\lceil}
\newcommand{\rc}{\rceil}
\newcommand{\lf}{\lfloor}
\newcommand{\rf}{\rfloor}

\newcommand{\tb}{\textbf}

\def\Wolg{Without loss of generality}

\begin{document}
\title{Optimal binary linear locally repairable codes with disjoint repair groups}
\author{Jingxue Ma and Gennian Ge
\thanks{The research of G. Ge was supported by the National Natural Science Foundation of China under Grant Nos.11431003 and 61571310, Beijing Scholars Program, Beijing Hundreds of Leading Talents Training Project of Science and Technology, and Beijing Municipal Natural Science Foundation.}
\thanks{J. Ma is with the School of Mathematical Sciences, Zhejiang University, Hangzhou 310027, China (email: majingxue@zju.edu.cn).}
\thanks{G. Ge is with the School of Mathematical Sciences, Capital Normal University, Beijing 100048, China (e-mail: gnge@zju.edu.cn).}
}

\date{}\maketitle

\begin{abstract}

In recent years, several classes of codes are introduced to provide some fault-tolerance and guarantee system reliability in distributed storage systems, among which locally repairable codes (LRCs for short) play an important role. However, most known constructions are over large fields with sizes close to the code length, which lead to the systems computationally expensive. Due to this, binary LRCs are of interest in practice.

In this paper, we focus on binary linear LRCs with disjoint repair groups. We first derive an explicit bound for the dimension $k$ of such codes, which can be served as a generalization of the bounds given in \cite{Goparaju2014, Wang2017, Zeh2015}. We also give several new constructions of binary LRCs with minimum distance $d=6$ based on weakly independent sets and partial spreads, which are optimal with respect to our newly obtained bound. In particular, for locality $r\in \{2,3\}$ and minimum distance $d=6$, we obtain the desired optimal binary linear LRCs with disjoint repair groups for almost all parameters.

\medskip
\noindent {{\it Keywords and phrases\/}: Locally repairable codes, disjoint repair group, weakly independent set, partial spread, distributed storage systems.
}\\
\smallskip

\end{abstract}


\section{Introduction}\label{secintro}

Modern large scale distributed storage systems, such as data centers, store data in a redundant form to ensure reliability against node failures. The simplest and most commonly used technique is $3$-replication, which has clear advantages due to its simplicity and fast recovery from data failures. However, this strategy entails large storage overhead and is nonadaptive for modern systems supporting the ``Big Data" environment.

To achieve better storage efficiency, erasure codes are employed, such as Windows Azure \cite{5}, Facebook's Hadoop cluster \cite{14}, where the original data are partitioned into $k$ equal-sized fragments and then encoded into $n$ fragments $(n\ge k)$ stored in $n$ different nodes. It can tolerate up to $d-1$ node failures, where $d$ is the minimum distance of the erasure code. Particularly, the maximum distance separable (MDS) code is a kind of erasure code that attains the maximal minimum distance with respect to the Singleton bound and thus provides the highest level of fault tolerance for given storage overhead.

However, if a node fails, which is the most common failure scenario, we may recover it by looking at the information on any $k$ remaining nodes. This is a slow and time consuming recovery process since we have to read data from $k$ nodes when we only want to recover data on one node. There are two famous metrics in the literature to quantify the efficiency of recovering, which are the repair bandwidth and repair locality. In this work, we study codes with small repair locality.

The concept of codes with locality was introduced by Gopalan et al. \cite{Gopalan12}, Oggier and Datta \cite{7}, and Papailiopoulos et al.~\cite{10}. The $i$-th coordinate of a code has locality $r$ if it can be recovered by accessing at most $r$ other coordinates. In this paper, an $[n,k,d]$ linear code with all-symbol locality $r$ is denoted by an $[n,k,d;r]$ LRC. When $r\ll k$, it greatly reduces the disk I/O complexity for repair.

When considering the fault tolerance level, the minimum distance is a key metric for LRCs. Gopalan et al. \cite{Gopalan12} first derive the following upper bound for codes with information locality:
\begin{equation}\label{singletonbound}
d\leq n-k-\lc\frac{k}{r}\rc+2,
\end{equation}
which is also called as Singleton-like bound since it degenerates to classical Singleton bound when $r=k.$ Later, in \cite{2,9}, the bound (\ref{singletonbound}) is generalized to vector codes and nonlinear codes. Although it certainly holds for all LRCs, it is not tight in many cases. The tightness of the bound (\ref{singletonbound}) is studied in \cite{16,Wang1510}.

We say an LRC is $d$-optimal if it satisfies bound (\ref{singletonbound}) with equality for given $n$, $k$ and $r$. For the case $(r + 1)|n$, $d$-optimal LRCs are constructed explicitly in \cite{18} and \cite{15} by using Reed-Solomon codes and Gabidulin codes respectively. However, both constructions are built over a finite field whose size is an exponential function of the code length $n$. In \cite{Tamo1408}, for the same case $(r + 1)|n$, the authors construct a $d$-optimal code over a finite field of size sightly greater than $n$ by using ``good'' polynomials. This construction can be extended to the case $(r+1)\nmid n$ with the minimum distance $d\ge n-k-\lceil \frac{k}{r}\rceil+1$ which is at most one less than the bound (\ref{singletonbound}). In \cite{Barg2017,Tamo2016}, the authors generalized this idea to cyclic codes and algebraic geometry codes.

In order to consider multiple erasures in local repair, two different models for recovering the erasures in parallel were put forward independently by Wang et al. \cite{Wang1411} and Prakash et al. \cite{11}. Also, another model which was called a sequential approach was proposed in \cite{Prakash1406}. There are lots of other works devoted to the locality in the handling of multiple node failures, such as \cite{Balaji16071,8,Song2017,Tamo1406,Tamo1408,Wang16isit}.

For the convenience of computer hardware implementation, codes over small (especially binary) alphabets are of particular interest. In the binary case, the bound (\ref{singletonbound}) is not tight in almost all cases. In fact, Hao et al. \cite{Hao2016} prove that there are only 4 classes of $d$-optimal binary LRCs meeting the Singleton-like bound in the sense of equivalence of linear codes. A new bound taking field size into consideration is derived in \cite{CMbound}, which is called Cadambe-Mazumdar (C-M) bound,
\begin{equation}\label{alphabound}
k\leq \min_{t\in \mathbb{Z}^{+}}\big[tr+k_{\text{opt}}^{(q)}(n-(r+1)t,d)\big],
\end{equation}
where $k_{\text{opt}}^{(q)}(n,d)$ is the largest possible dimension of a code, for given field size $q$, code length $n$, and minimum distance $d$.

The C-M bound can be attained by binary simplex codes \cite{CMbound}. Later in \cite{Silberstein2015}, new binary LRCs constructed via anticodes are presented to attain the bound (\ref{alphabound}) with locality $r=2,3$. Based on different base codes (systematic binary codes), Huang et al. \cite{Huang2016} propose constructions of binary LRCs with $d=3,4$ and $5$ using phantom parity check symbols, some of which are shown to be optimal with respect to their minimum distance. Binary cyclic LRCs with optimal dimension for distances 2, 6 and 10 and locality 2, are presented in \cite{Goparaju2014}. Also, a class of binary LRCs with optimal dimension for parameters $r=2$ and $d=10$, are constructed in \cite{Zeh2015}. Besides, a class of binary LRCs with minimum distance 4 can be seen in \cite{Shahabinejad2016}. Also in \cite{Wang2017}, the authors construct a class of optimal binary LRCs with minimum distance at least 6. Another construction of binary LRCs with minimum distance at least 6 is given in \cite{Nam2017}, showing that some examples are optimal with respect to bound (\ref{alphabound}). In \cite{Fu2017}, the authors investigate the locality of MacDonald codes and generalized MacDonald codes, and propose some constructions of binary LRCs.

Most of the optimal constructions of LRCs above adopt a special structure, which is called disjoint repair groups (Definition \ref{DRG}). In this paper, we are concerned with the bound and constructions of binary LRCs with disjoint repair groups. We first derive an explicit bound for the dimension $k$ of such codes (Theorem \ref{theorembound}), which can be served as a generalization of the bound given in \cite{Wang2017}. We say a binary LRC is $k$-optimal if it satisfies the bound in Theorem \ref{theorembound} with equality for given $n$, $d$ and $r$. We also give two classes of $k$-optimal binary LRCs for general $r$. And our first class of $k$-optimal binary linear LRCs contains the construction given in \cite{Wang2017}. Moveover, for locality $r\in \{2,3\}$ and minimum distance $d=6$, we obtain $k$-optimal binary linear LRCs with disjoint repair groups for almost all parameters.

An outline of this paper is as follows. The next section is about some notations, definitions and some known results on partial spread. Section \ref{secup} is devoted to give an upper bound on the dimension of $[n,k,d;r]_2$-LRCs. In Section \ref{secgeneral}, we give some optimal constructions for $d=6$. Some discussions for binary LRCs with $d\geq8$ and concluding remarks are presented in the final section.

\section{Preliminaries}\label{secpre}

The following notations are fixed throughout this paper

\begin{itemize}
  \item Let $[n]=\{1,2,\cdots,n\}$ and $[a,b]=\{a,a+1,\cdots,b\}$ if $a\leq b$ are two integers.
  \item Let $q$ be a prime power, and $\mathbb{F}_q$ be the finite field with $q$ elements.
  \item Let $\mathbb{F}_{q}^{n}$ denote the vector space of dimension $n$ over $\mathbb{F}_q$.
  \item For any vector $\tb{v}=(v_1,\cdots,v_n)\in \mathbb{F}_{q}^{n}$, let ${\rm supp}(\tb{v})=\{i\in [n] | v_i\neq 0\}$ and ${\rm wt}(\tb{v})=|{\rm supp}(v)|.$ For a set $S\subseteq[n],$ define $\tb{v}|_S=(v_{i_1},\ldots,v_{i_{|S|}})$ where $i_j\in S$ for $1\le j\le |S|$ and $1\le i_1<\cdots<i_{|S|}\le n$.
  \item $d(\tb{u},\tb{v})$ stands for Hamming distance between any two vectors $\tb{u},\tb{v} \in \mathbb{F}_{q}^{n}$.
  \item Let $U\subseteq\mathbb{F}_{q}^{n}$ be a set and $\tb{v}$ be a vector. Define $d(U,\tb{v})=\min\{d(\tb{u},\tb{v})| \tb{u}\in U\}.$
  \item Let $I_n$ be the $n\times n$ identity matrix. And let $\tb{1}_n$ and $\tb{0}_{n}$ be the all-one and all-zero vectors, respectively.
\end{itemize}

Firstly, we give the concept of weakly independent set which has appeared in \cite{Gopalan2014}.

\begin{definition}
A set $T\subseteq \mathbb{F}$ is $\tau$-wise weakly independent over $\mathbb{F}_2\subseteq \mathbb{F}$ if no set $T'\subseteq T$ where $2\leq |T'|\leq \tau$ has the sum of its elements equal to zero.
\end{definition}

Next, we briefly introduce partial spread.

\begin{definition}
A partial $t$-spread of $\mathbb{F}_{q}^{m}$ is a collection $S = \{W_1, \cdots, W_l\}$ of $t$-dimensional subspaces of $\mathbb{F}_{q}^{m}$ such that $W_i \cap W_j = \{0\}$ for $1 \leq i < j \leq l$. The number $l$ is called the size of $S$. Moreover, we call $S$ maximal if it has the largest possible size. Particularly, if $\cup_{i=1}^{l}W_i = \mathbb{F}_{q}^{m}$, we simply say that $S$ is a $t$-spread.
\end{definition}

It is well-known that a $t$-spread of $\mathbb{F}_{q}^{m}$ exists if and only if $t$ divides $n$. Bounds for the sizes of maximum partial $t$-spreads were heavily studied in the past, not only because of its own interest, but also because it has a great connection to a special kind of $q$ subspace codes in (network) coding theory. Here the codewords are nonzero subspaces of $\mathbb{F}_{q}^{m}$. The most widely used distance measure for subspace codes is the so-called subspace distance $d_S(U,V):=\dim(U+V)-\dim(U\cap V)$. For $k \in [m]$, we denote by $A_q(m,k,d)$ the maximal cardinality of subspace codes over $\mathbb{F}_{q}^{m}$ with minimum subspace distance $d$, where we additionally assume that the dimension of the codewords are all $k$. With this notation, the size of a maximum partial $k$-spread in $\mathbb{F}_{q}^{m}$ is given by $A_q(m,k,2k)$.

Only few results about $A_q(m,k,2k)$ are known. Here we list some of them which will be used in this paper.

\begin{lemma}\cite{Etzion2011}\label{lemspread1}
For positive integers $1\leq k \leq m$ with $m\equiv r \pmod k$, we have$$A_q(m,k,2k)\geq \frac{q^m-q^{k}(q^r-1)-1}{q^k-1}.$$
\end{lemma}

\begin{lemma}\cite{beutelspacher1975}
For positive integers $1\leq k \leq m$ with $m\equiv r \pmod k$, we have$$A_q(m,k,2k)=\left\{\begin{array}{lrc} \frac{q^m-1}{q^k-1} & \text{if } r=0, \\[1mm] \frac{q^m-q^{k+1}+q^k-1}{q^k-1} & \text{if } r=1. \end{array}\right.$$
\end{lemma}

\begin{corollary}
$$A_2(m,2,4)=\left\{\begin{array}{lrc} \frac{2^m-1}{3} & \text{if } m\equiv 0 \pmod 2,\\[1mm] \frac{2^m-5}{3} & \text{if } m\equiv 1 \pmod 2. \end{array}\right. $$
\end{corollary}

\begin{lemma}\cite{El2010}
$$A_2(m,3,6)=\left\{\begin{array}{lrc} \frac{2^m-1}{7} & \text{if } m\equiv 0 \pmod 3,\\[1mm]
\frac{2^m-9}{7} & \text{if } m\equiv 1 \pmod 3,\\[1mm]
\frac{2^m-18}{7} & \text{if } m\equiv 2 \pmod 3.
\end{array}\right. $$
\end{lemma}

Now, we give the definition of linear LRCs discussed in this paper.
\begin{definition}
The $i$-th coordinate, $1\leq i\leq n$, of an $[n,k,d]_q$ linear code $\mathcal{C}$ is said to have $r$-locality if it can be recovered by accessing at most $r$ other code symbols. Equivalently, there exists a codeword $\tb{h}_i$ in the dual code $\mathcal{C}^{\bot}$ such that $i\in {\rm supp}(\tb{h}_i)$ and $\mid{\rm supp}(\tb h_i)\mid\leq r+1.$
\end{definition}

\begin{definition}
An $[n,k,d]_q$ linear code $\mathcal{C}$ is said to have $r$-locality if each of the $i$-th coordinate has $r$-locality, where $1\leq i\leq n.$ We denote it by $[n,k,d;r]_q$.
\end{definition}

\begin{definition}\label{DRG}
We say an $[n,k,d;r]_2$ LRC has disjoint repair groups if there exists a set of local parity checks $\tb h_1, \tb h_2, \cdots, \tb h_l \in \mathcal{C}^{\bot}$ such that $\cup_{i=1}^{l}{\rm supp}(\tb h_i)=[n]$, ${\rm wt}(\tb h_i)=r+1$ and ${\rm supp}(\tb h_i)\cap{\rm supp}(\tb h_j)=\varnothing$ for $1\leq i\neq j\leq l.$
\end{definition}

\section{Upper bound for binary LRCs with disjoint repair groups}\label{secup}

In this section, we will give an upper bound for binary LRCs with disjoint repair groups based on Johnson bound. Let $\mathcal{C}$ be an $[n,k,d;r]_2$ LRC with disjoint repair groups. We assume that the parity check matrix $\mathbf{H}$ of $\mathcal{C}$ consists of two parts
\begin{equation}\label{parityh}
\mathbf{H}=\left(\begin{array}{lrc}\mathbf{H}_L\\\mathbf{H}_G\end{array}\right).
\end{equation}
The matrix $\mathbf{H}_L$ is designed to guarantee the locality of the code $\mathcal{C}$, and the matrix $\mathbf{H}_G$ determines the minimum distance of $\mathcal{C}$. Assume $r+1$ divides $n$ since we only consider the code with disjoint repair groups in this paper. \Wolg , we have
$$\mathbf{H}_L=I_{\frac{n}{r+1}} \bigotimes \tb 1_{r+1}.$$
Note that the sum of the rows of $\mathbf{H}_L$ is an all-one vector, it follows that the minimum distance of $\mathcal{C}$ must be even.

\begin{theorem}\label{theorembound}
For any $[n,k,d;r]_2$ binary linear LRC with disjoint repair groups, let $d=2t+2$ and $n=(r+1)l.$
\begin{enumerate}
  \item If $t+1$ is an odd integer, we have
  \begin{equation}\label{boundodd}
  k\leq \frac{rn}{r+1}-\Big\lceil \log_2 (\sum_{0\leq i_1+\cdots+i_l\leq \lfloor\frac{d-1}{4}\rfloor} \prod_{j=1}^{l} \binom{r+1}{2 i_j}) \Big\rceil.
  \end{equation}
  \item If $t+1$ is an even integer, we have
  \begin{equation}\label{boundeven}
  k\leq \frac{rn}{r+1}-\Big\lceil \log_2 (\sum_{0\leq i_1+\cdots+i_l\leq \lfloor\frac{d-1}{4}\rfloor} \prod_{j=1}^{l} \binom{r+1}{2 i_j}+\frac{\sum\limits_{i_1+\cdots+i_l=\frac{d}{4}} \prod\limits_{j=1}^{l}\binom{r+1}{2 i_j}}{\lfloor\frac{n}{t+1}\rfloor}) \Big\rceil.
  \end{equation}
\end{enumerate}
\end{theorem}

\begin{proof}
Let $\mathcal{C}$ be an $[n,k,d;r]_2$ binary linear LRC with disjoint repair groups, and $\{\tb h_1, \tb h_2, \cdots, \tb h_l\}$ be the set of local parity checks corresponding to the disjoint repair groups of $\mathcal{C}$. Put $L={\rm span}_{\mathbb{F}_2}\{\tb h_1, \tb h_2, \cdots, \tb h_l\}$ and $V=L^{\bot}$. Clearly, $L\subseteq \mathcal{C}^{\bot}$, $\dim(L)=l$, we arrive at $\ml{C}\subseteq V$ and $\dim(V)=n-l=\frac{rn}{r+1}.$

Noticing that $t$ is the packing radius of $\ml{C}$, the spheres of radius $t$ centered at any of the codewords are disjoint. Let $B_{V}(\tb c,t)=\{\tb v\in V | d(\tb c,\tb v)\leq t\},$ and denote $B_{V}(\tb 0,t)$ by $B_{V}(t).$ We obtain $|B_{V}(\tb c,t)|=|B_{V}(t)|$ since $\ml{C}\subseteq V.$ Let $\ml{N}$ be the vectors at distance $t+1$ from $\ml{C}$, so $\ml{N}=\{\tb x \in V | d(\ml{C},\tb x)=t+1\}.$ Clearly, we have

\begin{equation}\label{equbound}
|\ml{C}|\times|B_{V}(t)|+|\ml{N}| \leq 2^{\dim(V)}.
\end{equation}

We first compute the size of $B_{V}(t).$ Note that the linear space $L$ has weight enumerator polynomial $W_{L}(x,y)=(x^{r+1}+y^{r+1})^l.$ By the MacWilliams equality, the weight enumerator polynomial of $V$ is

$$W_{V}(x,y)=\frac{1}{|L|}W_{L}(x+y,x-y)=\sum_{0\leq u\leq \frac{n}{2}}A_{u}x^{n-2u}y^{2u},$$
where $A_u=\sum\limits_{i_1+\cdots+i_l=u} \prod\limits_{j=1}^{l}\binom{r+1}{2 i_j}.$ It follows that

\begin{equation}\label{equsph}
|B_{V}(t)|=A_0+\cdots+A_{\lf\frac{d-1}{4} \rf}=\sum_{0\leq i_1+\cdots+i_l\leq \lfloor\frac{d-1}{4}\rfloor} \prod_{j=1}^{l} \binom{r+1}{2 i_j}.
\end{equation}

The next step is to show a lower bound of $|\ml{N}|$.

When $t+1$ is odd, it's easy to get $|\ml{N}|=0$. In fact, if there exists $\tb x\in \ml{N}\subseteq V$ such that $d(\tb c,\tb x)=t+1$ for some $\tb c\in \ml{C}\subseteq V$, then we have ${\rm wt}(\tb c-\tb x)=d(\tb c,\tb x)=t+1$. This is impossible since for any $\tb v\in V$, ${\rm wt}(\tb v)$ is an even integer, which can be seen from the weight enumerator polynomial $W_{V}(x,y).$ Thus Eq.(\ref{boundodd}) can be derived directly from Eq.(\ref{equbound}) and Eq.(\ref{equsph}).

When $t+1$ is even, let $\Omega=\{(\tb c,\tb x)\in \ml{C}\times\ml{N} | d(\tb c,\tb x)=t+1\}$ and $\Omega_{c}=\{\tb x\in \ml{N} | (\tb c,\tb x)\in\Omega\}.$ Then $|\Omega|=\sum_{\tb c\in\ml{C}}|\Omega_{\tb c}|.$ Fix $\tb c\in \ml{C}.$ Let $\tb x\in V$ be a vector at distance $t+1$ from $\tb c$, i.e. ${\rm wt}(\tb c-\tb x)=t+1$. There are exactly $A_{\frac{d}{4}}$ such vectors $\tb x$.

We claim that such vectors $\tb x$ lie in $\Omega_{\tb c}.$ Because ${\rm wt}(\tb c-\tb x)=t+1,$ we know $d(\ml{C},\tb x)\leq t+1.$ Let $\tb c^{\prime}\in\ml{C}$ with $\tb c^{\prime}\neq \tb c.$ Then by the triangle inequality, we have $d\leq {\rm wt}(\tb c^{\prime}-\tb c)={\rm wt}(\tb c^{\prime}-\tb x-(\tb c-\tb x))\leq {\rm wt}(\tb c^{\prime}-\tb x)+{\rm wt}(\tb c-\tb x)=\text{wt}(\tb c^{\prime}-\tb x)+t+1,$ implying $\text{wt}(\tb c^{\prime}-\tb x)\geq t+1,$ yielding $d(\ml{C},\tb x)=t+1$ since $\tb c^{\prime}\in \ml{C}$ was chosen arbitrarily and we saw previously that $d(\ml{C},\tb x)\leq t+1.$ Therefore all such $\tb x$ lie in $\Omega_{\tb c}$ giving $|\Omega_{\tb c}|=A_{\frac{d}{4}}.$ Hence
\begin{equation}\label{equomega1}
|\Omega|=|\ml{C}|\times A_{\frac{d}{4}}.
\end{equation}

On the other hand, fix $\tb x\in\ml{N}.$ We get an upper bound for the number of $\tb c\in\ml{C}$ with $d(\tb c,\tb x)=t+1.$ Clearly, the set $\{\tb c-\tb x| \tb c\in\ml{C} \text{ with } d(\tb c,\tb x)=t+1\}$ is a binary constant weight code of length $n$ with codewords of weight $t+1$ and minimum distance $d=2t+2.$ Thus for each $\tb x\in\ml{N},$ there are at most $\lf\frac{n}{t+1}\rf$ choices for $\tb c$ with $d(\tb c,\tb x)=t+1.$ Hence
\begin{equation}\label{equomega2}
|\Omega|\leq |\ml{N}|\times \Big\lfloor\frac{n}{t+1}\Big\rfloor.
\end{equation}

Thus Eq.(\ref{boundeven}) can be derived directly from Eq.(\ref{equbound}), Eq.(\ref{equsph}), Eq.(\ref{equomega1}) and Eq.(\ref{equomega2}).
\end{proof}

\begin{remark}
Recently, Eq.(\ref{boundodd}) was derived in \cite{Wang2017}, our result can be viewed as a generalization.
\end{remark}

\begin{remark}
Our bound is tight when $d=4$. Eq.(\ref{boundeven}) degenerates to $k\leq \frac{rn}{r+1}-\lc\log_2 (1+r)\rc.$ There exists an $[n,k=\frac{rn}{r+1}-\lc\log_2 (1+r)\rc,d=4;r]_2$ LRC with disjoint repair groups, which can be found in \cite{Shahabinejad2016}.
\end{remark}

\begin{remark}
In \cite{Goparaju2014}, binary LRCs with disjoint repair groups are constructed from primitive cyclic codes for special parameters such as $r=2$, $d=2,6,10$ and $n=2^m-1$ with $m$ even. Also binary LRCs with parameters $r=2$, $d=10$ and $n=2^m+1$ with $m$ odd are constructed in \cite{Zeh2015}. Furthermore, there exists an $[n=\frac{2^s-1}{2^t-1},k\geq\frac{rn}{r+1}-s,d\geq 6;r=2^t]_2$ LRC with disjoint repair groups, which can be found in \cite{Wang2017}. All of the constructions above attain our bound in Eq.(\ref{boundodd}).
\end{remark}

At the end of this section, we will give a general construction for binary LRCs with disjoint repair groups, which is asymptotically optimal when $n$ approaches the infinity. Let $S^{\prime}=\{\beta_1, \beta_2, \cdots, \beta_n\}\subseteq \mathbb{F}_{2^{\lc \log_2 n\rc}}$, where $\beta_1, \beta_2, \cdots, \beta_n$ are distinct elements of the field. Consider $S=\{\mathbf{a}_1, \mathbf{a}_2, \cdots, \mathbf{a}_n\}\subseteq \mathbb{F}_{2^{\lc \log_2 n\rc}}^{t},$ where $\mathbf{a}_i=(\beta_{i}, \beta_{i}^{3}, \cdots, \beta_{i}^{2t-1})^{\top}$, $i\in [n]$.

\begin{construction}\label{congener}
Define a binary LRC $\ml{C}$ with parity check matrix $\mathbf{H}$ given in Eq.(\ref{parityh}). Let $\mathbf{H}_L=I_{\frac{n}{r+1}} \bigotimes \tb 1_{r+1}$ and $\mathbf{H}_G$ be a $t\lc \log_2 n\rc \times n$ matrix whose columns are binary expansions of the vectors $\{\mathbf{a}_1, \mathbf{a}_2, \cdots, \mathbf{a}_n\}.$
\end{construction}

\begin{theorem}
The code $\ml{C}$ obtained from Construction \ref{congener} is an $[n,k\geq\frac{rn}{r+1}-t\lc\log_2 n\rc,d\geq 2t+2;r]_2$ LRC with disjoint repair groups.
\end{theorem}

\begin{proof}
As the number of rows of $\mathbf{H}$ is $\frac{n}{r+1}+t\lc \log_2 n\rc$, we have $k\geq\frac{rn}{r+1}-t\lc\log_2 n\rc$. It follows from $\mathbf{H}_{L}$ that the locality is $r$.

Now we prove the lower bound on $d$. Note that the minimum distance of $\mathcal{C}$ must be even since the sum of the rows of $\mathbf{H}_L$ is an all-one vector. Therefore it suffices to show that $d\geq 2t+1$. Suppose there exists a codeword $\tb c\in\ml{C}$ such that $\mathbf{H}\tb c^{\top}=\tb 0$ and $2\leq {\rm wt}(\tb c)\leq 2t.$ \Wolg , we may assume that $\tb c=(c_1,\cdots,c_{2t},0,\cdots,0)$\textbf, which implies that $\sum_{i=1}^{2t} c_i \tb a_i=0.$ Thus we get $\sum_{i=1}^{2t} c_i\beta_{i}^{2s-1}=0,$ for $s\in [t]$. Raising both sides to the $2^b$-th power, we obtain $\sum_{i=1}^{2t} c_i\beta_{i}^{2^{b}(2s-1)}=0,$ where $b$ is a nonnegative integer. Notice that for each nonzero integer $m$, it can be uniquely written as $m=2^{s}e$ with $e$ odd. We have already obtained $\sum_{i=1}^{2t} c_i\beta_{i}^{m}=0,$ where $m\geq 1$ is an integer. In other words, we get a nonzero solution $x=(c_1\beta_{1},c_2\beta_{2},\cdots,c_{2t}\beta_{2t})^{\top}$ to the system $\mathbf{M}x=0,$ where
$$\mathbf{M}=\left(
    \begin{array}{cccc}
      1 & 1 & \cdots & 1 \\
      \beta_1 & \beta_2 & \cdots & \beta_{2t} \\
      \vdots & \vdots & \ddots & \vdots \\
      \beta_{1}^{2t-1} & \beta_{2}^{2t-1} & \cdots & \beta_{2t}^{2t-1} \\
    \end{array}
  \right)
$$
is the Vandermonde matrix, which has full rank. We get a contradiction.
\end{proof}

\section{$k$-optimal constructions for binary LRCs with disjoint repair groups}\label{secgeneral}

In this section, we will give constructions for $k$-optimal $[n,k,d;r]_2$ LRCs with disjoint repair groups. The parity check matrix $\mathbf{H}$ can be represented as follows:
\begin{equation}\label{conhr}
\mathbf{H}=\left(\begin{array}{lrc}\mathbf{H}_L\\\mathbf{H}_G\end{array}\right)=\left(\begin{array}{cccc}
\mathbf{H}_L^1   &  \mathbf{H}_L^2   & \ldots  &  \mathbf{H}_L^{l}\\[1mm]
\mathbf{H}_G^1   &  \mathbf{H}_G^2   & \ldots  &  \mathbf{H}_G^{l}
\end{array}
\right),
\end{equation}
where $l=\frac{n}{r+1}.$ For $i\in [l],$ $\mathbf{H}_L^i$ is an $l\times(r+1)$ matrix whose $i$-th row is $\tb{1}_{r+1}$ and the other rows are all zero\textbf, $\mathbf{H}_G^i$ denotes the $i$-th $(n-k-l)\times(r+1)$ sub-matrix of $\mathbf{H}_G$.


The sub-matrix $\mathbf{H}_G=(\mathbf{H}_G^1~\mathbf{H}_G^2~\ldots~\mathbf{H}_G^{l})$ will determine the minimum distance of the code $\mathcal{C}$. It is well known that the minimum distance of a linear code is at least $d$ if and only if any $d-1$ columns of $\mathbf{H}$ are linearly independent.

\begin{lemma}\label{lemd}
Let $\mathcal{C}$ be a code that is defined by the parity check matrix $\mathbf{H}$ in Eq.(\ref{conhr}), where $l=\frac{n}{r+1},$ and $t\geq 0$ be an integer. Then $d\geq 2t+2$ holds if and only if
$$\sum_{i=1}^{l}\sum_{j=1}^{a_i}\tb c_{j}^{i}\neq \tb 0,$$
where $a_1,a_2,\cdots,a_l$ satisfy the following two conditions: (1) for $1\leq i\leq l$, $0\leq a_i\leq\min\{2t,r+1\}$ is an even integer; (2) $2\leq\sum_{i=1}^{l}a_i\leq 2t$. And $\{\tb c_{1}^{i},\tb c_{2}^{i},\cdots,\tb c_{a_i}^{i}\}$ is a collection of any $a_i$ columns from $\mathbf{H}_G^i$.
\end{lemma}

\begin{proof}
Let $\mathbf{H}^{(i)}=\left(\begin{array}{lrc}\mathbf{H}_{L}^{i}\\[1mm]\mathbf{H}_{G}^{i}\end{array}\right)$ be the $i$-th block of $\mathbf{H},$ and $\tb h_{j}^{i}$ be the column of $\mathbf{H}^{(i)}$ such that $\tb h_{j}^{i}|_S=\tb c_{j}^{i},$ where $S=[l+1,n-k].$

We first prove the necessity. Since $d\geq 2t+2,$ we know that any $\leq d-1$ columns of $\mathbf{H}$ must be linearly independent. Then for any $a_1,a_2,\cdots,a_l$ satisfying the two conditions (1) and (2), we have $\sum_{i=1}^{l}\sum_{j=1}^{a_i}\tb h_{j}^{i}\neq \tb 0.$ Note that any even number of columns of $\mathbf{H}_{L}^{i}$ sum to $\tb0$, it follows that $\sum_{i=1}^{l}\sum_{j=1}^{a_i}\tb c_{j}^{i}\neq \tb 0.$

For the sufficiency, one only needs to show that any $2t+1$ columns of $\mathbf{H}$ are linearly independent. Then we will show that no $m$ columns of $\mathbf{H}$ sum to $\tb 0,$ for $m\in [2t+1].$ Let $m=\sum_{i=1}^{l}a_i$ and $\{\tb h_{1}^{i},\tb h_{2}^{i},\cdots,\tb h_{a_i}^{i}\}$ be a collection of any $a_i$ columns from $\mathbf{H}^{(i)}$, where $0\leq a_i\leq\min\{2t+1,r+1\}$ is an integer, for $1\leq i\leq l$. We need to show $\sum_{i=1}^{l}\sum_{j=1}^{a_i}\tb h_{j}^{i}\neq \tb 0.$

If $a_j$ is odd for some $j\in[l]$, then the $j$-th coordinate of $\sum_{i=1}^{l}\sum_{j=1}^{a_i}\tb h_{j}^{i}$ is 1, that is $\sum_{i=1}^{l}\sum_{j=1}^{a_i}\tb h_{j}^{i}\neq \tb 0.$ Then we can assume that $a_j$ is even for all $j\in[l].$ Note that $\sum_{i=1}^{l}\sum_{j=1}^{a_i}\tb h_{j}^{i}|_S=\sum_{i=1}^{l}\sum_{j=1}^{a_i}\tb c_{j}^{i}\neq \tb 0,$ it follows that $\sum_{i=1}^{l}\sum_{j=1}^{a_i}\tb h_{j}^{i}\neq \tb 0.$ This completes the proof.
\end{proof}

\subsection{$k$-optimal constructions for $d=6$}

In order to construct optimal binary LRCs with $d\geq6$, by Lemma \ref{lemd}, we can easily derive the following corollary.
\begin{corollary}\label{coro6}
Let $\mathcal{C}$ be a code that is defined by the parity check matrix $\mathbf{H}$ in Eq.(\ref{conhr}), where $l=\frac{n}{r+1}$. Then $d\geq 6$ holds if and only if the columns of $\mathbf{H}_G$ satisfy the following conditions
\begin{itemize}
\itemsep=0pt \parskip=0pt
\item[(1)] $\tb c_{1}^{i}+\tb c_{2}^{i}\neq \tb 0$ for each $i\in [l]$.\\
\item[(2)] $\tb c_{1}^{i}+\tb c_{2}^{i}+\tb c_{3}^{i}+\tb c_{4}^{i} \neq \tb 0$ for each $i\in [l]$.\\
\item[(3)] $\tb c_{1}^{i}+\tb c_{2}^{i}+\tb c_{1}^{j}+\tb c_{2}^{j} \neq \tb 0$ for $i\neq j\in [l]$.
\end{itemize}
\end{corollary}

Now we make a brief description of our construction. First, given space $W$, we want to construct a set of vectors of $W$ satisfying (1) and (2), which can be constructed by weakly independent set. Second, let $W_i$ be the vector space generated by the columns of $\mathbf{H}_G^i$, $i\in [l]$, it follows that $\dim(W_i)\leq n-k-l$. In order to ensure condition (3), we can suppose $W_i\cap W_j=\{\tb 0\}$. Such spaces can be selected from a (partial) spread.

\begin{lemma}\label{lem4a1}
Let $V$ be a $t$-dim subspace of $\mathbb{F}_2^m$, and $\{\tb e_1,\tb e_2,\cdots,\tb e_t\}$ be a basis of $V$. Suppose there exists a binary linear code with parameters $[n,n-t,d\geq 5]$ with parity check matrix $\mathbf{H}=(\tb h_1, \tb h_2,\cdots,\tb h_n)$. Then the set $T=\{\tb 0\}\cup\{\tb f_i| i\in [n]\},$ where $\tb f_i=\sum_{j\in {\rm supp}(\tb h_i)}\tb e_j,$ is  $4$-wise weakly independent over $\mathbb{F}_2.$
\end{lemma}

\begin{proof}
It suffices to show that for any $\{i_1,\ldots,i_r\} \subseteq [n]$ and $1\leq r \leq4$, $\sum_{j=1}^{r} \tb f_{i_j}\neq \tb 0,$ since $\tb 0$ doesn't affect the sum. \Wolg, we can assume that $\sum_{i=1}^{r}\tb f_i=0,$ for $1\leq r \leq4.$ Then $\sum_{j=1}^{r}\tb h_j=\tb 0,$ which indicates that $\{\tb h_1,\tb h_2,\ldots,\tb h_r\}$ are linearly dependent. Since $\mathbf{H}$ is the parity check matrix of the code with $d\geq 5$, we infer that any $\leq 4$ columns of $\mathbf{H}$ are linearly independent. This leads to a contradiction. So $T$ must be $4$-wise weakly independent over $\mathbb{F}_2.$
\end{proof}

Below, we give two classes of $k$-optimal binary LRCs for general $r$.

\subsubsection{\tb {The first class of optimal binary LRCs}}

\begin{lemma}\cite{Macwilliams1977, Wang2017}\label{lem4a2}
There exists a binary linear code with parameters $[2^m,2^m-2m,\geq 5]$ for each integer $m\geq3$.
\end{lemma}

\begin{construction}\label{conr}
Let $r=2^t$, and $\{W_1, W_2, \cdots, W_a\}$ be a maximum partial $2t$-spread of $\mathbb{F}_{2}^{s}$. We denote a basis of $W_i$ by $\{\tb e_{1}^{(i)},\tb e_{2}^{(i)},\cdots,\tb e_{2t}^{(i)}\}.$ When $t\geq3$, by Lemma \ref{lem4a2}, there exists a code with parameters $[2^t,2^t-2t,\geq 5]_2$. Let $T^{(i)}$, for $i\in[a]$, be the set given by Lemma \ref{lem4a1}. When $t=1, 2$, we define $T^{(i)}=\{\tb 0, \tb e_{1}^{(i)},\tb e_{2}^{(i)},\cdots,\tb e_{2t}^{(i)}\}.$ Let $\mathbf{H}_G^i$ be an $s \times (2^t+1)$ matrix whose columns are the vectors in $T^{(i)}.$ Then we can define a binary LRC $\ml{C}$ with parity check matrix $\mathbf{H}$ given in Eq.(\ref{conhr}), where $\frac{s}{r}<l\leq a$.
\end{construction}

\begin{theorem}\label{thm4.1}
The code $\ml{C}$ obtained from Construction \ref{conr} is an $[n=(2^t+1)l,k\geq\frac{rn}{r+1}-s, d\geq6;r=2^t]_2$ LRC. Moreover, when
$$\frac{2^{s-1}-1}{2^{t-1}(2^t+1)} < l\leq A_2(s,2t,4t),$$ we have $k=\frac{rn}{r+1}-s, d=6$, which is optimal with respect to bound (\ref{boundodd}).
\end{theorem}

\begin{proof}
First, in view of the definition of parity check matrix $\mathbf{H}$, we get $n=(2^t+1)l, k\geq\frac{rn}{r+1}-s$ and locality $r=2^t$. Then we need to verify $d\geq6$. It suffices to show that the matrix $\mathbf{H}_G$ satisfies the conditions in Corollary \ref{coro6}. It follows from Lemma \ref{lem4a1} that $\mathbf{H}_G$ satisfies conditions (1) and (2). Notice that for any two vectors $\tb c_{1}^{i}, \tb c_{2}^{i}$ in $\mathbf{H}_{G}^{i}$, $\tb c_{1}^{i}+\tb c_{2}^{i}$ is a nonzero vector in subspace $W_i$. Then $\tb c_{1}^{i}+\tb c_{2}^{i}$ and $\tb c_{1}^{j}+\tb c_{2}^{j}$ are two distinct vectors since $W_i \cap W_j=\{\tb 0\},$ for $i\neq j.$ So condition (3) is satisfied.

When $\frac{2^{s-1}-1}{2^{t-1}(2^t+1)} < l\leq A_2(s,2t,4t)$, we now show the optimality of the code $\ml{C}$. In fact, for an $[n,k,d\geq 6;r]_2$ binary linear LRC with disjoint repair groups, we know $k\leq \frac{rn}{r+1}-\lc\log_2(1+\frac{r}{2}n)\rc$ by our bound in Eq.(\ref{boundodd}). Since $\frac{2^{s-1}-1}{2^{t-1}(2^t+1)} < l\leq A_2(s,2t,4t),$ we have $2^{s-1}<  1+\frac{rn}{2}\leq 2^s.$ It follows that $\lc\log_2(1+\frac{r}{2}n)\rc=s,$ which implies $k\leq \frac{rn}{r+1}-s.$ So we get $k=\frac{rn}{r+1}-s,$ which is optimal with respect to bound (\ref{boundodd}). Finally, in this case, we have to show that $d=6.$ In fact, if $d>6$, we have $d\geq 8$ since the minimum distance of $\ml{C}$ must be even. Notice that $1+\frac{rn}{2}+\frac{\binom{l}{2}\binom{r+1}{2}\binom{r+1}{2}}{n/4} > 2^s.$ Then by our bound in Eq.(\ref{boundeven}), we obtain $k<\frac{rn}{r+1}-s$, which contradicts the previous conclusion $k=\frac{rn}{r+1}-s.$
\end{proof}

\begin{remark}
Taking $2t |s$, $s\geq 4t$ and $l=A_2(s,2t,4t)=\frac{2^s-1}{2^{2t}-1}$ in Construction \ref{conr}, we get the $k$-optimal binary linear LRCs constructed in \cite{Wang2017}.
\end{remark}

\begin{example}\label{exgr2}
Let $s\equiv u \pmod {2t}$, we know $A_2(s,2t,4t)\geq\frac{2^s-2^{2t}(2^u-1)-1}{2^{2t}-1}$ by Lemma \ref{lemspread1}. It follows that whenever $\frac{2^{s-1}-1}{2^{t-1}(2^t+1)} < l\leq \frac{2^s-2^{2t}(2^u-1)-1}{2^{2t}-1}$, we can always construct an $[n=(2^t+1)l,k=\frac{rn}{r+1}-s, d=6;r=2^t]_2$ LRC, which is optimal with respect to bound (\ref{boundodd}) by Construction \ref{conr}. For each pair of $r$ and $s$, the above construction contains a number of optimal examples. (See Figure 1 below for an illustration.)

\begin{figure}[h]
\center
\begin{tabular}{|c|c|c|c|}
\hline
$s$  &    $l$       &$n$     &$k$      \\\hline
$4$  &$[3,5]$     &$3l$    &$2l-4$   \\\hline
$5$  &$[6,9]$     &$3l$    &$2l-5$   \\\hline
$6$  &$[11,21]$    &$3l$    &$2l-6$   \\ \hline
$7$  &$[22,41]$  &$3l$    &$2l-7$   \\ \hline
\end{tabular}\\
\caption{Optimal binary LRCs with disjoint repair groups for $r=2, d=6.$}
\end{figure}
\end{example}

\subsubsection{\tb {The second class of optimal binary LRCs}}

\begin{lemma}\cite{Macwilliams1977}\label{lem4a3}
There exists a binary linear code with parameters $[2^t+2^{\lf(t+1)/2\rf}-1,2^t+2^{\lf(t+1)/2\rf}-2t-2, 5]$, where $t\geq3$.
\end{lemma}

\begin{construction}\label{conrs}
Let $r=2^t+2^{\lf(t+1)/2\rf}-1$, and $\{W_1, W_2, \cdots, W_a\}$ be a maximum partial $(2t+1)$-spread of $\mathbb{F}_{2}^{s}$. Denote a basis of $W_i$ by $\{\tb e_{1}^{(i)},\tb e_{2}^{(i)},\cdots,\tb e_{2t+1}^{(i)}\}$. When $t\geq3$, by Lemma \ref{lem4a3}, there exists a binary code with parameters $[2^t+2^{\lf(t+1)/2\rf}-1,2^t+2^{\lf(t+1)/2\rf}-2t-2, 5]$. Let $T^{(i)}$, for $i\in[a]$, be the set given by Lemma \ref{lem4a1}. When $t=1, 2$, we define $T^{(i)}=\{\tb 0, \tb e_{1}^{(i)},\tb e_{2}^{(i)},\cdots,\tb e_{2t+1}^{(i)}\}.$ Let $\mathbf{H}_G^i$ be an $s \times (2^t+1)$ matrix whose columns are the vectors in $T^{(i)}.$ Then we can define a binary LRC $\ml{C}$ with parity check matrix $\mathbf{H}$ given in Eq.(\ref{conhr}), where $\frac{s}{r}<l\leq a$.
\end{construction}

\begin{theorem}\label{thm4.2}
The code $\ml{C}$ obtained from Construction \ref{conrs} is an $[n=(r+1)l,k\geq\frac{rn}{r+1}-s, d\geq6;r=2^t+2^{\lf(t+1)/2\rf}-1]_2$ LRC. Moreover, when $$\frac{2^{s}-2}{(2^t+2^{\lf(t+1)/2\rf}-1)(2^t+2^{\lf(t+1)/2\rf})} < l\leq A_2(s,2t+1,4t+2),$$
we have $k=\frac{rn}{r+1}-s, d=6$, which is optimal with respect to bound (\ref{boundodd}).
\end{theorem}

\begin{proof}
The argument is analogous to that in Theorem \ref{thm4.1}. We omit it here.
\end{proof}

\begin{example}
Taking $t=3,$ we obtain a binary linear code with parameters $[11,4,5]$. Let $\{W_1, W_2, \cdots, W_{129}\}$ be a $7$-spread of $\mathbb{F}_{2}^{14}$. Put $125\leq l\leq 129$. Then we obtain an $[n=12l, k=11l-14, d=6; r=11]_2$ LRC by Theorem \ref{thm4.2}, which is optimal with respect to bound (\ref{boundodd}).
\end{example}

\begin{example}\label{exreq3}
Taking $t=1,$ we have $r=3.$ Then we obtain an optimal $[n=4l, k=3l-s, d=6; r=3]_2$ LRC by Theorem \ref{thm4.2} whenever $\frac{2^{s}-2}{12} < l\leq A_2(s,3,6)$. (See Figure 2 below for an illustration.)
\begin{figure}[h]
\center
\begin{tabular}{|c|c|c|c|}
\hline
$s$  &    $l$       &$n$     &$k$      \\\hline
$6$  &$[6,9]$     &$4l$    &$3l-6$   \\\hline
$7$  &$[11,17]$     &$4l$    &$3l-7$   \\\hline
$8$  &$[22,34]$    &$4l$    &$3l-8$   \\ \hline
\end{tabular}\\
\caption{Optimal binary LRCs with disjoint repair groups for $r=3, d=6.$}
\end{figure}
\end{example}

%
%

\subsection{Almost complete constructions for optimal LRCs with  $r\in\{2,3\}$}

When taking $t=1$, according to Construction \ref{conr} (or Construction \ref{conrs}), we could obtain $k$-optimal binary LRCs with $r=2$ (or $r=3$). It seems that for any integer $l$, we can always get optimal LRCs with parameters $[3l,2l-s,6;2]$, except for one class of $l$ (see Example \ref{exgr2}). In fact, we could prove the following consequences, that is, Lemma \ref{lem413} and Theorem \ref{thm414}. However, in the case of $r=3$, for many values of $l$, we can't obtain $k$-optimal LRCs by means of Construction \ref{conrs} (see Example \ref{exreq3}). In this subsection, we will give a different construction for the case $r=3$ by modifying Construction \ref{conr} slightly, which could cover almost all integers $l$ except for only one class.

We first give a comprehensive analysis of the above optimal binary linear LRCs with parameters $[n,k,6;2]_2$. Without loss of generality, we set $A_2(3,2,4):=1$. Now we define
$$N_m:=[A_2(2m-1,2,4)+2,A_2(2m+1,2,4)].$$
So we have
\begin{lemma}\label{lem413}
The set $\cup_{m=2}^{\infty} N_m$ covers all the positive integers larger than 2, except for $\frac{2^{2m+1}-2}{3}.$
\end{lemma}

\begin{proof}
Notice that $N_m=[A_2(2m-1,2,4)+2,A_2(2m+1,2,4)]=[\frac{2^{2(m-1)+1}-2}{3}+1,\frac{2^{2m+1}-2}{3}-1]$, the conclusion is clear.
\end{proof}

\begin{theorem}\label{thm414}
Let $n=3l$ such that $l\neq \frac{2^{2m+1}-2}{3}$, where $m\geq 2$ is an integer. Then there exists an $[n,k,6; 2]_2$ binary linear LRC with parameter
$$k=\left\{\begin{array}{cc} 2l-2m & \text{if } l\in [A_2(2m-1,2,4)+2,A_2(2m,2,4)],\\[1mm]
2l-2m-1 & \text{if } l\in [A_2(2m,2,4)+1,A_2(2m+1,2,4)],
\end{array}\right.$$
which is optimal with respect to our bound in Eq.(\ref{boundodd}).
\end{theorem}

\begin{proof}
It can be derived directly from Theorem \ref{thm4.1} by taking $t=1$ and $s=2m$ or $2m+1.$
\end{proof}

\begin{example}\label{exr2}
Let $l=4\in [A_2(3,2,4)+2,A_2(4,2,4)]\subseteq N_2$. Then we have $2m=4$. We first construct a $2$-spread in $\mathbb{F}_2^4$. Let $\alpha$ be a primitive element of $\mathbb{F}_{2^4}$, and $\{1,\alpha,{\alpha}^2,{\alpha}^3\}$ form a basis for $\mathbb{F}_{2^4}$ over $\mathbb{F}_2.$ Set $\beta=\alpha^5$, thus we have a $2$-spread
$\Big\{W_i={\rm span}_{\mathbb{F}_2}\{\alpha^i,\alpha^i\beta\} | 0\leq i\leq 4\Big\}.$ Without loss of generality, we choose the first $4$  subspaces of them. Thus we obtain an optimal $[12,4,6;2]_2$ LRC with parity check matrix
\[\mathbf{H}=
\left(\begin{array}{cccccccccccc}
1 & 1 & 1 & 0 & 0 & 0 & 0 & 0 & 0 & 0 & 0 & 0 \\
0 & 0 & 0 & 1 & 1 & 1 & 0 & 0 & 0 & 0 & 0 & 0 \\
0 & 0 & 0 & 0 & 0 & 0 & 1 & 1 & 1 & 0 & 0 & 0 \\
0 & 0 & 0 & 0 & 0 & 0 & 0 & 0 & 0 & 1 & 1 & 1 \\
0 & 0 & 0 & 0 & 0 & 1 & 0 & 0 & 1 & 0 & 1 & 0 \\
0 & 0 & 1 & 0 & 0 & 1 & 0 & 1 & 0 & 0 & 0 & 1 \\
0 & 0 & 1 & 0 & 1 & 0 & 0 & 0 & 1 & 0 & 0 & 0 \\
0 & 1 & 0 & 0 & 0 & 0 & 0 & 0 & 1 & 0 & 0 & 1
\end{array}\right).
\]
\end{example}

\begin{remark}
For the case $l=\frac{2^{2m+1}-2}{3}$, note that $A_2(2m+1,2,4)<l<A_2(2m+2,2,4)$. Taking $s=2m+2$ and $t=1$ in Construction \ref{conr}, we can construct an $[n=2^{2m+1}-2,k\geq\frac{2^{2m+2}-4}{3}-(2m+2),d\geq 6;r=2]_2$ LRC, which is nearly optimal, since $k\leq \frac{2^{2m+2}-4}{3}-2m-1$ by our bound in Eq.(\ref{boundodd}).
\end{remark}

Now we give an almost complete construction for optimal binary LRCs with $r=3$. We modify the parity check matrix $\mathbf{H}$ of $r=2$ by adding a column and a row in each $\mathbf{H}^{(i)}.$

\begin{construction}\label{conr3}
Let $\{W_1, W_2, \cdots, W_a\}$ be a maximum partial $2$-spread of $\mathbb{F}_{2}^{s}$. Denote a basis of $W_i$ by $\{\tb e_{1}^{(i)},\tb e_{2}^{(i)}\}$. Then we can define a binary LRC $\ml{C}$ with parity check matrix $\mathbf{H}$ given in Eq.(\ref{conhr}) for $\frac{s+1}{3}<l\leq a$, where the submatrices $\mathbf{H}_G^i, i\in[l]$ are defined as follows:
$$\mathbf{H}_G^i=\left(\begin{array}{cccc}
\tb 0   &  \tb e_{1}^{(i)}   & \tb e_{2}^{(i)}  &  \tb e_{1}^{(i)}+\tb e_{2}^{(i)}\\[1mm]
1   &  0   & 0  &  0
\end{array}\right).$$
\end{construction}

\begin{remark}\label{remarkr3}
An argument similar to that of Theorem \ref{thm4.1} shows that the code $\ml{C}$ obtained from Construction \ref{conr3} is an $[n=4l,k\geq 3l-s-1, d\geq6;r=3]_2$ LRC.
\end{remark}

The following theorem says that for almost all $n$ that can be divided by $4$, we have a $k$-optimal LRC.

\begin{theorem}
Let $n=4l$ such that $l\neq \frac{2^{2m+1}-2}{3}$, where $m\geq 2$ is an integer. Then there exists an $[n=4l,k,d=6;r=3]_2$ binary linear LRC with parameter
$$k=\left\{\begin{array}{cc} 3l-2m-1 & \text{if } l\in [A_2(2m-1,2,4)+2,A_2(2m,2,4)],\\[1mm]
3l-2m-2 & \text{if } l\in [A_2(2m,2,4)+1,A_2(2m+1,2,4)],
\end{array}\right.$$
which is optimal with respect to our bound in Eq.(\ref{boundodd}).
\end{theorem}

\begin{proof}
{\bf Case 1: $l\in [A_2(2m-1,2,4)+2,A_2(2m,2,4)]$}.

In this case, we can take $s=2m$ in Construction \ref{conr3}, then we know that the code $\ml{C}$ obtained from Construction \ref{conr3} is an $[n=4l,k\geq 3l-2m-1, d\geq6;r=3]_2$ LRC by Remark \ref{remarkr3}. On the other hand, we get $k\leq 3l-\lc\log_2(1+6l)\rc=3l-2m-1$ by our bound in Eq.(\ref{boundodd}). Thus we have $k= 3l-2m-1,$ which attains the bound in Eq.(\ref{boundodd}).

Next, we will show $d=6.$ In fact, if $d>6$, we have $d\geq 8$ since the minimum distance of $\ml{C}$ must be even. After a simple computation, we obtain $k<3l-2m-1$ by our bound in Eq.(\ref{boundeven}), which contradicts the previous conclusion $k= 3l-2m-1.$

{\bf Case 2: $l\in [A_2(2m,2,4)+1,A_2(2m+1,2,4)]$}.

We can take $s=2m+1$ in Construction \ref{conr3}. The remainder of the argument is similar to that used in {\bf Case 1}. We omit it here. This completes the proof.
\end{proof}

\begin{example}
By Example \ref{exr2} and Construction \ref{conr3}, we could construct an optimal $[16,7,6;3]_2$ LRC with parity check matrix
\[\mathbf{H}=
\left(
  \begin{array}{cccccccccccccccc}
    1 & 1 & 1 & 1 & 0 & 0 & 0 & 0 & 0 & 0 & 0 & 0 & 0 & 0 & 0 & 0 \\
    0 & 0 & 0 & 0 & 1 & 1 & 1 & 1 & 0 & 0 & 0 & 0 & 0 & 0 & 0 & 0 \\
    0 & 0 & 0 & 0 & 0 & 0 & 0 & 0 & 1 & 1 & 1 & 1 & 0 & 0 & 0 & 0 \\
    0 & 0 & 0 & 0 & 0 & 0 & 0 & 0 & 0 & 0 & 0 & 0 & 1 & 1 & 1 & 1 \\
    0 & 0 & 0 & 0 & 0 & 0 & 1 & 1 & 0 & 0 & 1 & 1 & 0 & 1 & 0 & 1 \\
    0 & 0 & 1 & 1 & 0 & 0 & 1 & 1 & 0 & 1 & 0 & 1 & 0 & 0 & 1 & 1 \\
    0 & 0 & 1 & 1 & 0 & 1 & 0 & 1 & 0 & 0 & 1 & 1 & 0 & 0 & 0 & 0 \\
    0 & 1 & 0 & 1 & 0 & 0 & 0 & 0 & 0 & 0 & 1 & 1 & 0 & 0 & 1 & 1 \\
    1 & 0 & 0 & 0 & 1 & 0 & 0 & 0 & 1 & 0 & 0 & 0 & 1 & 0 & 0 & 0 \\
  \end{array}
\right).
\]
\end{example}

\begin{remark}
The above example also attains the C-M bound (\ref{alphabound}). Furthermore, if we take 5 subspaces from a $2$-spread in $\mathbb{F}_2^4$, we can construct a $[20,10,6;3]_2$ LRC, which also attains the C-M bound (\ref{alphabound}).
\end{remark}

\begin{remark}
For the case $l=\frac{2^{2m+1}-2}{3}$, note that $A_2(2m+1,2,4)<l<A_2(2m+2,2,4)$. Taking $s=2m+2$ in Construction \ref{conr3}, we can construct an $[n=\frac{4 (2^{2m+1}-2)}{3},k\geq2^{2m+1}-2m-5,d\geq 6;r=3]_2$ LRC, which is nearly optimal, because $k\leq 2^{2m+1}-2m-4$ by our bound in Eq.(\ref{boundodd}).
\end{remark}

\section{Discussions and concluding remarks}\label{sec8}

In this section, we will discuss how to construct an optimal binary linear LRC with $d\geq8$. We assume $r=2$ for convenience. It is straightforward to get the following corollary by Lemma \ref{lemd}.

\begin{corollary}\label{coro8}
Let $\mathcal{C}$ be a code that is defined by the parity check matrix $\mathbf{H}$ in Eq.(\ref{conhr}), where $n=3l$. Then $d\geq 8$ holds if and only if the columns of $\mathbf{H}_G$ satisfy the following conditions
\begin{itemize}
\itemsep=0pt \parskip=0pt
\item[(1)] $\tb c_{1}^{u}+\tb c_{2}^{u}\neq \tb 0$ for each $u\in [l]$.\\
\item[(2)] $\tb c_{1}^{u}+\tb c_{2}^{u}+\tb c_{1}^{v}+\tb c_{2}^{v} \neq \tb 0$ for $1 \leq u<v\leq l$.\\
\item[(3)] $\tb c_{1}^{u}+\tb c_{2}^{u}+\tb c_{1}^{v}+\tb c_{2}^{v}+\tb c_{1}^{w}+\tb c_{2}^{w} \neq \tb 0$ for $1 \leq u<v<w\leq l$.
\end{itemize}
\end{corollary}

Clearly, conditions (1) and (2) can be guaranteed by a maximum partial $2$-spread $S$. In order to ensure condition (3), we define property P as follows. A subset $Y$ of $S$ is said to have the property P if any three distinct subspaces of $Y$, denoted by $W_i, W_j, W_k$, satisfy $\dim(W_i+W_j+W_k)=6$. So condition (3) could be satisfied if we have such a $Y$.

\begin{construction}\label{cond8}
Let $S=\{W_1, W_2, \cdots, W_a\}$ be a maximum (partial) $2$-spread of $\mathbb{F}_{2}^{s}$. Let $Y$ have the property P with maximum possible size $b$. Write $Y=\{W_{y_1}, W_{y_2}, \cdots, W_{y_b}\}.$ Denote a basis of $W_{y_i}$ by $\{\tb e_{1}^{(i)},\tb e_{2}^{(i)}\}$. Then we can define an $[n=3l, k\geq 2l-s, d\geq8; r=2]_2$ binary LRC $\ml{C}$ with parity check matrix $\mathbf{H}$ given in Eq.(\ref{conhr}) for $\frac{s}{2}<l\leq b$, where the submatrices $\mathbf{H}_G^i, i\in[l]$ are defined as follows:
$$\mathbf{H}_G^i=\left(\begin{array}{ccc}
\tb 0   &  \tb e_{1}^{(i)}   & \tb e_{2}^{(i)}
\end{array}\right).$$
\end{construction}

\begin{example}
Let $\alpha$ be the primitive element of $\mathbb{F}_{2^6},$ and its minimal polynomial be $x^6+x^4+x^3+x+1$. Let $S$ be a $2$-spread of $\mathbb{F}_{2}^{6}$, we know $|S|=A_{2}(6,2,4)=21$. In fact, $S$ can be constructed as follows: $W_i={\rm span}_{\mathbb{F}_2}\{\alpha^i,\beta\alpha^i\}$ for $0\leq i\leq 20$, where $\beta=\alpha^{21}$. We have a possible set $Y=\{W_1, W_2, W_3, W_4, W_{10}, W_{19}\}$ by MAGMA. Taking $l=6$, we can derive an $[18,6,8;2]_2$ LRC with the parity check matrix
\[\mathbf{H}=
\left(\begin{array}{cccccccccccccccccc}
1&1&1&0&0&0&0&0&0&0&0&0&0&0&0&0&0&0 \\
0&0&0&1&1&1&0&0&0&0&0&0&0&0&0&0&0&0 \\
0&0&0&0&0&0&1&1&1&0&0&0&0&0&0&0&0&0 \\
0&0&0&0&0&0&0&0&0&1&1&1&0&0&0&0&0&0 \\
0&0&0&0&0&0&0&0&0&0&0&0&1&1&1&0&0&0 \\
0&0&0&0&0&0&0&0&0&0&0&0&0&0&0&1&1&1 \\
1&0&1&0&0&0&0&0&0&0&1&1&1&1&0&0&1&1 \\
0&1&1&1&0&1&0&0&0&0&1&1&0&0&0&1&0&1 \\
0&1&1&0&1&1&1&0&1&0&0&0&0&0&0&1&1&0 \\
0&1&1&0&1&1&0&1&1&1&1&0&0&1&1&1&0&1 \\
0&0&0&0&1&1&0&1&1&0&0&0&1&1&0&0&1&1 \\
0&0&0&0&0&0&0&1&1&0&1&1&1&0&1&1&0&1
\end{array}\right),
\]
which is optimal with respect to the bound (\ref{boundeven}).
\end{example}

\begin{remark}
If we take $l=5$ in the example above, we get a binary linear LRC with parameters $[15,4,8;2]_2$. Moreover, we can obtain binary linear LRCs with parameters $[20,8,8;3]_2$ and $[24,11,8;3]_2$ by using the method in Construction \ref{conr3}. All of these examples are optimal with respect to bound (\ref{boundeven}).
\end{remark}

For $d=6$, by means of $4$-weakly independent sets and (partial) spreads, we construct optimal binary LRCs. For $d=8$, we propose a method to obtain optimal binary LRCs. That is, we need to construct a set at first, whose elements are certain subspaces in the (partial) spread. And this set should have the property that the dimension of the sum of arbitrary three subspaces in this set equals the sum of the dimensions of the three subspaces (for $r=2,3$). For bigger $r$, if we want to get $k$-optimal LRCs, we may also need to consider $6$-weakly independent sets. In addition, it would be interesting to give constructions for general $d$, which match the bound in Theorem \ref{theorembound}. We will consider this problem as a future work.

\end{document}